\documentclass{amsart}

\usepackage{amsmath,graphicx,amsfonts,lineno}

\newtheorem{theorem}{Theorem}[section]

\newtheorem{lemma}[theorem]{Lemma}

\newtheorem{conjecture}[theorem]{Conjecture}

\def\QED{\ensuremath{{\square}}}

\newcommand{\col}[1]{#1_\mathrm{col}}
\def\R{\mathop{I\!\!R}\nolimits}

\begin{document}


\title{On the number of radial orderings of planar point sets}

\author{J.M.~D\'{i}az-B\'{a}\~{n}ez}
\address{\newline 
Departamento Matem\'{a}tica Aplicada II. \newline
 Universidad de Sevilla \newline 
Seville, Spain}
\email{dbanez@us.es}

\author{R.~Fabila-Monroy}
\address{\newline Departamento de
  Matem\'aticas \newline Centro de Investigación y de Estudios Avanzados del Instituto Politécnico Nacional.
 \newline Mexico City,
  Mexico} \email{ruyfabila@math.cinvestav.edu.mx}

\author{P.~P\'{e}rez-Lantero}
\address{\newline Escuela de Ingenier\'ia Civil en Inform\'atica.\newline
 Universidad de Valpara\'iso \newline Valparaiso, Chile} \email{pablo.perez@uv.cl}

\thanks{J.M.D.-B. is suppported by project FEDER MEC MTM2009-08652 and ESF EUROCORES programme EuroGIGA - ComPoSe IP04 - MICINN Project EUI-EURC-2011-4306.\\ P.P.-L. is partially supported by project FEDER MEC MTM2009-08652 and grant FONDECYT 11110069.\\ R.F.-M partially supported by Conacyt of Mexico, grant 153984.}
\date{\today}

\maketitle

\begin{abstract}
Given a set $S$ of $n$ points in the plane, a  \emph{radial ordering} of $S$ with respect to a point $p$ (not in $S$) is a clockwise circular ordering of the elements in $S$ by angle around $p$. If $S$ is two-colored, a  \emph{colored radial ordering} is a radial ordering of $S$ in which only the colors of the points are considered. In this paper, we obtain bounds on the number of distinct non-colored and colored radial orderings of $S$. We assume a strong general position on $S$, not three points are collinear and not three lines---each passing through a pair of points in $S$---intersect in a point of $\R^2\setminus S$. In the colored case, $S$ is a set of $2n$ points partitioned into $n$ red and $n$ blue points, and $n$ is even. We prove that: the number of distinct radial orderings of $S$ is at most $O(n^4)$ and at least $\Omega(n^3)$; the number of colored radial orderings of $S$ is at most $O(n^4)$ and at least $\Omega(n)$; there exist sets of points with $\Theta(n^4)$ colored radial orderings and sets of points with only $O(n^2)$ colored radial orderings.
\end{abstract}


\section{Introduction}

Let $S$ be a set of $n$  points in the plane. 
We say that $S$ is in \emph{strong general position} if it is in \emph{general position}
(not three of its points are collinear) and every
time that three lines---each passing through a pair of points in $S$---intersect, they
do so in a point in $S$. 
Unless otherwise noted, all point sets in this paper are in strong general position.
Let $p$ be  a point not in $S$ such
that $S \cup \{p\}$ is in general position; we call $p$ an \emph{observation point}.
A \emph{radial ordering} of $S$ with respect to $p$ is 
a clockwise circular ordering of the points in $S$ by their angle around $p$.
Thus these orderings are equivalent under rotations.
If every point in $S$ is assigned one of two colors, say red and blue, then
 a \emph{colored radial ordering} of $S$ with respect to $p$ is 
a circular clockwise ordering of the colors of the points in $S$ by their angle around $p$.
Thus permutations between points
of the same color yield the same colored radial ordering.

Let $\rho(S)$ be the number of distinct radial orderings of $S$
with respect to every observation point in the plane. Likewise,
let $\col{\rho}(S)$ be the  number of distinct colored radial orderings of $S$
with respect to every observation point in the plane. We define the following functions:

\begin{align*}
    f(n)&:=\max\{\rho(S): S \text{ is a set of } n \text{ points}\}  \\
    \col{f}(n)&:=\max\{\col{\rho}(S): n \text{ is even and } S \text{ is a set of } n \text{ red and } n \text{ blue points}\} \nonumber \\
    g(n)&:=\min\{\rho(S): S \text{ is a set of } n \text{ points}\}  \\
    \col{g}(n)&:=\min\{\col{\rho}(S): n \text{ is even and } S \text{ is a set of } n \text{ red and } n \text{ blue points}\} \\
\end{align*}

In this paper we prove the following bounds.

\begin{equation*}
  \begin{array}{ccccc}
    &&f(n)&=&\Theta(n^4) \\
    &&\col{f}(n)&=&\Theta(n^4)\\

    \Omega(n^{3}) &\le &g(n)&\le&O(n^4)\\
    \Omega(n) &\le & \col{g}(n)& \le& O(n^2)\\
  
 \end{array}
\end{equation*}

The first equality ($f(n)=\Theta(n^4)$) has been noted
before in the literature. In \cite{ferran,poly,robots_original,numberof} $f(n)=O(n^4)$ is
proved. In \cite{numberof} the author proves both $f(n)=O(n^4)$
and $f(n)=\Omega(n^4)$. As far as we know,
all the other bounds are new.

A different problem but in the same setting
has been studied recently in \cite{robots_original}. In that paper, the authors
study what  a robot can infer from its environment
when all the information that is available
is the cyclic positions of some landmarks
as they appear from the robot's position.
Other authors  have considered problems of the same flavor,
when a similar kind of information is available.
See for example \cite{counting,minimalist,simple}. In \cite{ferran} the authors
study the algorithmic problem of updating the radial ordering of
a moving observation point.

We point out that computing the radial ordering of $S$
around every point \emph{in} $S$ is an unavoidable
step in some geometric algorithms, as for example, performing a radial sweeping of a point
set. Moreover, many optimization problems are solved by considering the arrangement generated
by every line passing through every pair of points in $S$, and finding the optimum point inside each of the $O(n^4)$ cells in the arrangement \cite{optimization}. In many cases this is because  the radial ordering of the points in $S$ around every point within a cell is the same. It could be interesting in this scenario to know how many cells induce the same radial ordering.

For a bi-colored point set, a radial sweeping algorithm also requires the ordering as an initial step, so it could be useful to know bounds on the number of different colored radial ordering of $S$ from points in the plane. From the combinatorial point of view, this problem is related to partitioning  bi-colored point sets with \emph{$k$-fans} \cite{4-fan,k-fans}. A $k$-fan in the plane is a point $p$ (called the center) and $k$ rays emanating from $p$.
This structure can be used to partition $S$  into $k$ monochromatic subsets and it depends
only on the colored radial ordering of $S$ with respect to $p$. The existence of balanced---each
part having an equal number of red an blue points---$k$-fans 
for colored point sets has been studied in recent papers \cite{bereg05,bereg00} 
but, as far as we know, the number of different monochromatic 
partitions induced by $k$-fans has not yet been considered. 

The assumptions that $S$ is in strong general position;
that $S$ has the same number
of red and blue points and that $n$ is even, may
seem arbitrary. However, the three of them are crucial
hypothesis in our results (see Section \ref{sec:con}).

A preliminary version
of this paper appeared in \cite{ECG2011}.

\section{Uncolored Case}\label{sec:uncolored}

\begin{figure}[t]
  \begin{center}
    \includegraphics[width=0.4\textwidth]{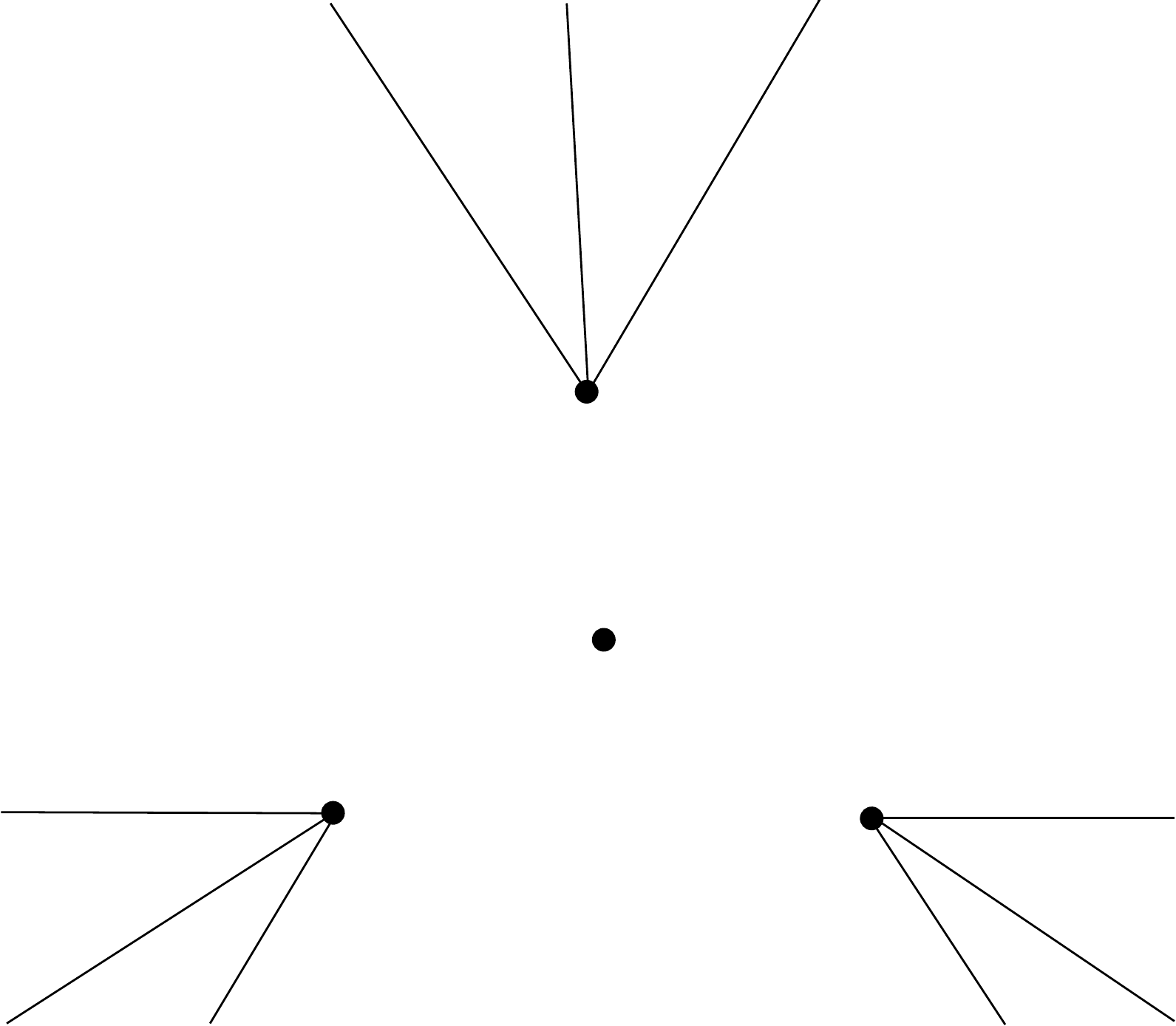}
  \end{center}
    \caption{The order partition of a set of four points.}\label{fig:order_partition}
\end{figure}

In this section $S$ is a set of $n$ points
in strong general position in the plane.
We discretize the problem by partitioning
the set of observation points into a finite 
number of sets so that points in the same set
induce the same radial ordering.
This partition  is made by half-lines, which
if crossed by an observation point, generate a transposition of two
consecutive elements in
the radial ordering. For every
pair of points $x_1,x_2 \in S$, consider the line passing through them.
Contained in this line we have two half-lines;
one  begins in $x_1$ and does not contain
$x_2$, while the other begins in $x_2$ and
does not contain $x_1$. Two observation points are in the same element of the partition if they can
be connected by a path which does not intersect any half-line. 
We call this partition the \emph{order partition} (see Figure \ref{fig:order_partition}). 
Since it induces a decomposition of the plane,
we refer to its elements as \emph{cells}. The order partition
is used (under different names) also in \cite{ferran,poly,robots_original}.
Note that if a point moves in a path not
crossing any half-line, the radial ordering with respect to this point
is the same throughout the motion. Thus points in the same cell induce the 
same radial ordering. As a set of three points already shows,
the converse is not true in general; two observation points may lie in different
cells of the order partition and induce the same radial ordering of $S$.

As mentioned before, the following two bounds on $f(n)$ have been
proved before; we provide proofs 
for completeness.

\begin{theorem} \label{thm:f_upper}
$f(n) \le O(n^4)$
\end{theorem}
\begin{proof}
The order partition cannot have more cells than the 
arrangement induced by the lines passing through each pair of points in $S$.
Such an arrangement has $O(n^4)$ cells.
\end{proof}

\begin{theorem}\label{thm:f_lower}
$f(n)\ge \Omega(n^4)$
\end{theorem}
\begin{proof}
Follows from $f(n)\ge \col{f}(n)$ and Theorem \ref{thm:fcol_lower}
\end{proof}

We now prove an upper and a lower bound on $g(n)$. The upper bound
follows from the upper bound on $f(n)$ 
\begin{theorem}
$g(n) \le O(n^4)$
\end{theorem}
\begin{flushright} \QED \end{flushright}

The lower bound on $g(n)$ is far more elaborate. First
we show a lower bound of $\Omega(n^4)$ on the
number of cells in the order partition. As far as we know,
this is the first lower bound ever given on the size
of the order partition.

\begin{theorem}\label{thm:part_size}
The number of cells in the order partition of $S$ is $\Omega(n^4)$.
\end{theorem}
\begin{proof}
Let $\mathcal{L}$ be the set of lines passing through each pair 
of points in $S$. First we show that in the line arrangement $\mathcal{A}$
generated by $\mathcal{L}$ there are $\frac{n^4}{8}-O(n^3)$ cells.
We may regard $\mathcal{A}$ as a plane graph $G$ with an extra vertex $v^*$, such 
that $v^*$ is contained in all the unbounded faces of $G$ and in all the lines in 
$\mathcal{L}$. 
Let $V$, $E$ and $F$ be the number of vertices, edges and faces of $G$ respectively.
Note that by our strong general assumption $G$ has exactly:
one vertex of degree $n(n-1)$; $n$ vertices of degree $n-1$
and $V-n-1$ vertices of degree $4$. Thus $E$ equals
$n(n-1)/2+n(n-1)/2+2V-2n-2$ which is $2V+O(n^2)$.
By Euler's formula $F$ equals $V+O(n^2)$. Let $p$ and 
$q$ be any two points in $S$. Let $L_p$ be the set of lines
in $\mathcal{L}$ that contain $p$ but not $q$. Let  $L_q$ be
the set of lines in $\mathcal{L}$ that contain $q$ but not $p$.
Note that each line in  $L_p$ intersects each line in $L_q$ in
a point not in $S$, with the exception of only two cases.
The case in which both lines are parallel or the case
in which they both contain the same point in $S\setminus \{p,q\}$. 
In total each case occurs at most $(n-2)$ times; 
thus $p$ and $q$ induce at least $(n-2)(n-2)-2(n-2)=(n-2)(n-4)$ vertices of degree $4$.
Doing this for every pair of vertices, we count each
degree $4$ vertex exactly four times. In total there are
$\frac{n^4}{8}-O(n^3)$ vertices and $\frac{n^4}{8}-O(n^3)$ faces
in $G$ (and the same number of cells in $\mathcal{A}$). The order partition can be obtained
by removing from $\mathcal{A}$ each line segment joining  a pair
of points of $S$. Let $e_1,\dots,e_{\binom{n}{2}}$ be these
line segments in any given order.
Let $\operatorname{cr}(S)$ be the number of pairs of these edges
that intersect in their interior. Let $M$ be the number 
of vertices of $\mathcal{A}$ lying in the interior of any of the $e_i$'s. Note that for 
each set of four points of $S$ we obtain: one vertex
of $M$ if the set is in convex position and three vertices
otherwise. The number of sets of four elements of $S$
that are in convex position is precisely $\operatorname{cr}(S)$.
Thus $M=3\binom{n}{4}-2\operatorname{cr}(S)$. It is known
that $\operatorname{cr}(S)$ is bounded from below by $\frac{3}{8}\binom{n}{4}$
(see \cite{bernardo,gelasio,lovasz}). Thus $M$ is at most 
$\frac{9}{4}\binom{n}{4}=\frac{3}{32}n^4-O(n^3)$.
We remove each $e_i$ in order, and show that at the end 
$\Omega(n^4)$ cells remain. Let $d_i$ be the number
of vertices in $M$ lying in the interior of
$e_i$ just before it is removed. Note that when $e_i$
is removed, in the worst case, $d_i+1$ cells 
of $\mathcal{A}$ are lost.
Thus when all of the $e_i$'s  are removed at least
$\frac{n^4}{8}-O(n^3)-\sum_{i=1}^{\binom{n}{2}}(d_i+1)=\frac{n^4}{8}-O(n^3)-M-\binom{n}{2}\ge\frac{n^4}{32}-O(n^3)$
cells remain.
\end{proof}

We now prove a useful lemma for finding distinct radial orderings
of $S$.

\begin{lemma}\label{lem:sep} (Partition Lemma)
Let $(R,B)$ be a partition of $S$. Let $p$ and $q$ 
be two points in different cells of the order partition.
 If no half-line spanned by a point in $R$ and point in $B$ intersects the line segment
with endpoints $p$ and $q$, then the radial
orderings of $S$ as seen from $p$ and $q$ are distinct.
\end{lemma}
\begin{proof}
Since $p$ and $q$ lie in different
cells of the order partition, the line segment
joining them must intersect at least one half-line. Let
$x_1,x_2 \in S$ be the points
defining this half-line. Note that $x_1$ and $x_2$ 
are both in $R$ or both in $B$. Let $x_3 \in S$ be a point
in the element of the partition not containing $x_1$ and $x_2$.
Assume without loss of generality that
the radial ordering of $\{x_1, x_2, x_3\}$
with respect to $p$  is $[x_1, x_2, x_3]$.
Since  the half-line is crossed only once,
the radial ordering of $\{x_1,x_2,x_3\}$ with
respect to $q$ is $[x_2, x_1, x_3]$.
Therefore, the radial ordering of $S$ with respect to $p$ is different
from the radial ordering with respect to $q$.
\end{proof}

The following Lemma is also proved
in \cite{poly}(Theorem 4.2) and in \cite{robots_original}(Theorem~1), we include a proof for completeness.

\begin{lemma}\label{lem:convex_interior}
Let $p$ and $q$ be two observation points in the interior
of the convex hull of $S$, lying in different cells of the order
partition. Then the radial orderings of $S$ with respect to
$p$ and $q$ are distinct.
\end{lemma}
\begin{proof}
Let $\ell$ be the straight line containing $p$ and $q$.
Since $p$ and $q$ are in the interior of the convex hull of $S$,
$\ell$ partitions $S$ into two sets that together with $p$
and $q$ satisfy the conditions
of the Partition Lemma.
\end{proof}

\begin{figure}[t]
 \begin{center}
  \includegraphics[width=0.4\textwidth]{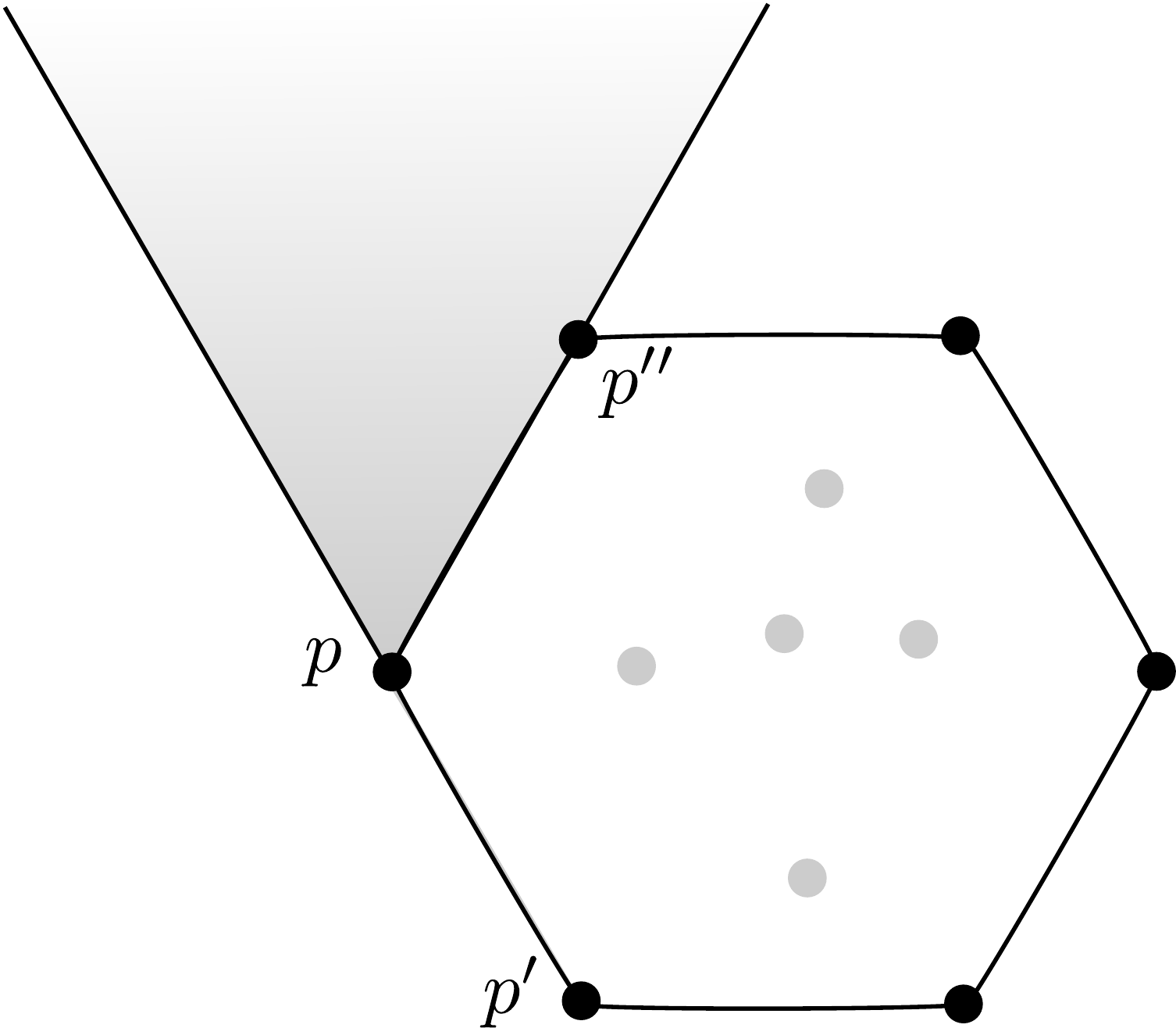}
  \end{center}
   \caption{The cone in the proof of Theorem \ref{thm:g_lower}.}\label{fig:cone}
\end{figure}

Finally, we combine the above results to prove
the lower bound on $g(n)$.

\begin{theorem}\label{thm:g_lower}
$g(n) \ge \Omega(n^{3})$
\end{theorem}
\begin{proof}
A cell of the order partition is \emph{interior} if it intersects
the interior of the convex hull of $S$ and it is
\emph{exterior} if it has a point not in the interior of
the convex hull of $S$. Note that a cell can be both interior and exterior.
If less than half of the the cells are exterior
then at least half of them are interior and 
we are done by Theorem \ref{thm:part_size} and 
Lemma \ref{lem:convex_interior}. 
Assume then,  that at least half of the cells are exterior. Thus
there are $\Omega(n^4)$ exterior cells.

Let $C$ be the convex hull of $S$ and $m$ be its
number of vertices.  Let $p$ be one of these vertices. Let $p'$ and $p''$ be the vertices
previous and next to $p$ in $C$
in clockwise order. Let $\Gamma_p$ be the convex cone
with apex $p$,
bounded by: the infinite ray with apex $p$ and 
passing through $p''$ and the infinite
ray with apex $p'$ and passing through $p$ (see Figure \ref{fig:cone}).
Let $R:=\{p\}$ and $B:=S \setminus \{p\}$.
Note that any two points in $\Gamma_p$
lying in different cells of the order partition,
together with $R$ and $B$, satisfy the conditions
of the Partition Lemma. Therefore, the radial orderings
of $S$ with respect to any two points in $\Gamma_p$  lying in different
cells of the order partition are distinct.
For each vertex of $C$ define
such a cone. Every exterior cell intersects one 
of these cones. Therefore there is a cone
intersecting $\Omega(n^4/m)=\Omega(n^{3})$
of them and the result follows. 
\end{proof}

\section{Colored Case} \label{sec:colored}

In this section $n$ is even, and $S$ is a set of $n$ red and $n$ blue points
in strong general position in the plane. We prove
upper and lower bounds on $\col{f}(n)$ and $\col{g}(n)$.

\begin{theorem}\label{thm:fcol_upper}
$\col{f}(n) \le O(n^4)$
\end{theorem}
\begin{proof}
Follows from $\col{f}(n) \le f(n)$ and Theorem \ref{thm:f_upper}
\end{proof}

\begin{figure}[t]
  \begin{center}
    \includegraphics[width=0.9\textwidth]{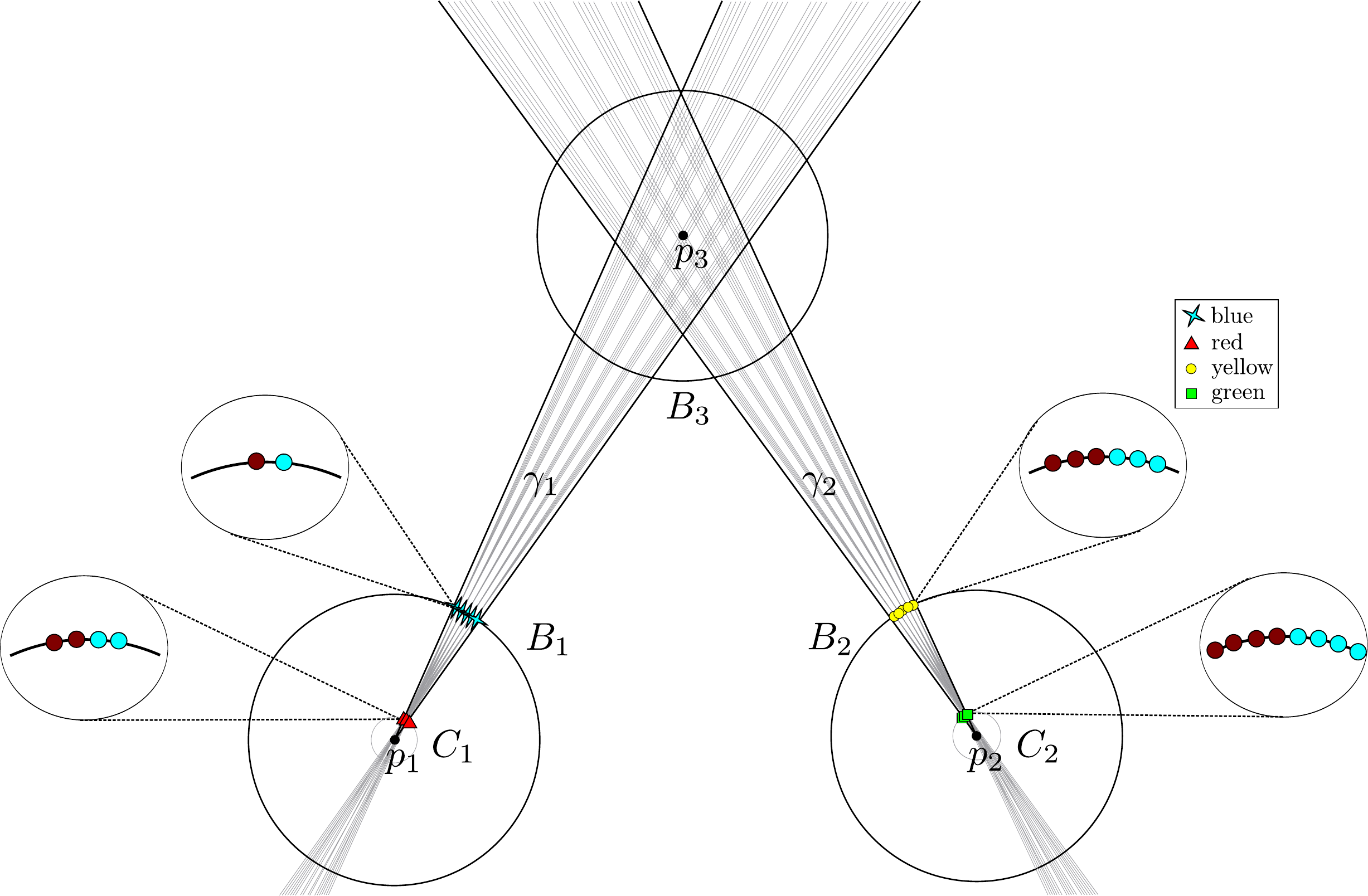}
  \end{center}
    \caption{A bi-colored set with $\Omega(n^4)$ different
colored radial orderings.}\label{fig:colormany}
\end{figure}

\begin{theorem}\label{thm:fcol_lower}
$\col{f}(n) \ge \Omega(n^4)$.
\end{theorem}
\begin{proof}
Assume that $n \ge 20$.
We start by constructing a four-colored
set of points $S'$ with $\Omega(n^4)$ distinct colored radial
orderings; afterwards we obtain $S$ by replacing each point 
of a given color with a suitable ``pattern'' of 
red and blue points. These four patterns are chosen
so that if they appear consecutively in a radial 
ordering, then any other equivalent radial ordering
must match patterns of the same type. Since
the patterns behave like the original four colors,
the new set also has $\Omega(n^4)$ colored
radial orderings.

Let $B_1$, $B_2$ and $B_3$ be
three disks of radius $1/4$, whose
centers $p_1$, $p_2$, and $p_3$ are
the vertices of an equilateral
triangle of side length equal to one.
Let  $\varepsilon, \alpha>0$. Let 
$C_1$ and $C_2$ be circles
of radius $\varepsilon$  centered
at $p_1$ and $p_2$, respectively.
Let $\gamma_1$ and $\gamma_2$ 
be infinite wedges of angle $\alpha$, with apices $p_1$ and $p_2$
respectively. Assume that $\gamma_1$ is
bisected by the line segment joining $p_1$ and $p_3$, while
$\gamma_2$ is bisected by the line segment joining $p_2$ and $p_3$. 
Refer to Figure \ref{fig:colormany}. Let $m$ and $r$ be the only natural numbers such
that $n=10m+r$ and $10 \le r \le 19$.
Partition $\gamma_1$ with $m$ infinite rays emanating from
$p_1$, such that the angle between two consecutive
rays is $\alpha/(m+1)$. Do likewise for $\gamma_2$, with
$m$ infinite rays emanating from $p_2$ .
At every point of intersection of these rays
with the boundary of $B_1$, place a blue point;
at every point of intersection with $C_1$
a red point; at every point of intersection
with the boundary of $B_2$ a yellow point;
finally at every point of intersection with $C_2$
a green point.  Thus $m$ points of each color are placed.
This ends the construction of $S'$.

 Let $L$ be the set of lines
passing through a red and a blue point. Let $L'$ be
the set of lines passing through a yellow and a green 
point. Choose $\alpha$ and $\varepsilon$  small enough so that the following
conditions are met: 
(1) Every line in $L \cup L'$ intersects the interior of $B_3$ and these
are the only lines passing through two points of $S'$ that
do. 
(2) No two lines in $L$ nor two lines in $L'$ intersect
at a point in the interior of  $B_3$.
(3) Every line in $L$ intersects every line in $L'$ at
a point in $B_3$. By the previous conditions
and the fact that $|L|=m^2$ and $|L'|=m^2$,  $L \cup L'$
splits $B_3$  into precisely
$(m+1)^4$ cells. For each of these cells
choose a point $q_i$ in its interior. We show that
the colored radial orderings of $S'$ as seen
from each of these points is different.
Note that for each point in $B_3$ there is a line 
separating the red and blue points
from the green and yellow points. Thus we may
assume, that the colored radial orderings as seen
from points in $B_3$ are written 
so that all the blue and red points appear before
the green and yellow points. Let $q_i$ and $q_j$ be
two points in different cells of the
order partition. Consider the colored radial ordering
when walking from $q_i$ to $q_j$ in a straight line. By conditions (1) and (2),
the only transpositions that occur when a half-line is crossed
 is between a red and a blue
point or between a yellow and a green point. This
implies that the $k$-th red point
is always the same red point and that the number of blue points after
the $k$-th red point is either
increasing or decreasing monotonically;
the same observation holds for the green
and yellow points. Therefore, in the walk
once a line in $L \cup L'$ is crossed, all colored
radial orderings afterwards will be distinct.
Thus the number of different colored radial orderings
of $S'$ is at least $(m+1)^4$, which is $\Omega(n^4)$.

 To construct $S$, we replace the points
in $S'$ by patterns of red and blue points, in such
a way that the colored radial orderings at points $q_i$ remain
different.
The points in the patterns replacing a point $p \in S'$ are placed
consecutively in the same circle containing $p$.
If these points are placed close enough to $p$, then
they will appear consecutively in the colored radial ordering with respect to
every point $q_i$.
The points of $S'$ are replaced  in the following way: every blue point with a pattern of one
red and one blue point; every red point with a pattern
of two red and two blue points; every yellow point
with a pattern of three red and three blue points; and
every green point with a pattern of four red
and four blue points. Refer to Figure \ref{fig:colormany}.
 Note that our choice
of patterns implies that two equivalent
radial orderings must match patterns of the same type.
So far, $10m$ red and $10m$ blue points have been
placed. The remaining $2r$ points
can be placed in such a way that in the radial
ordering with respect to every point $q_i$  these $r$ red points appear
consecutively followed by these $r$ blue points. This final condition
guarantees that the colored radial orderings at each $q_i$ remain
different.
\end{proof}



\begin{figure}[t]
  \begin{center}
    \includegraphics[width=0.6\textwidth]{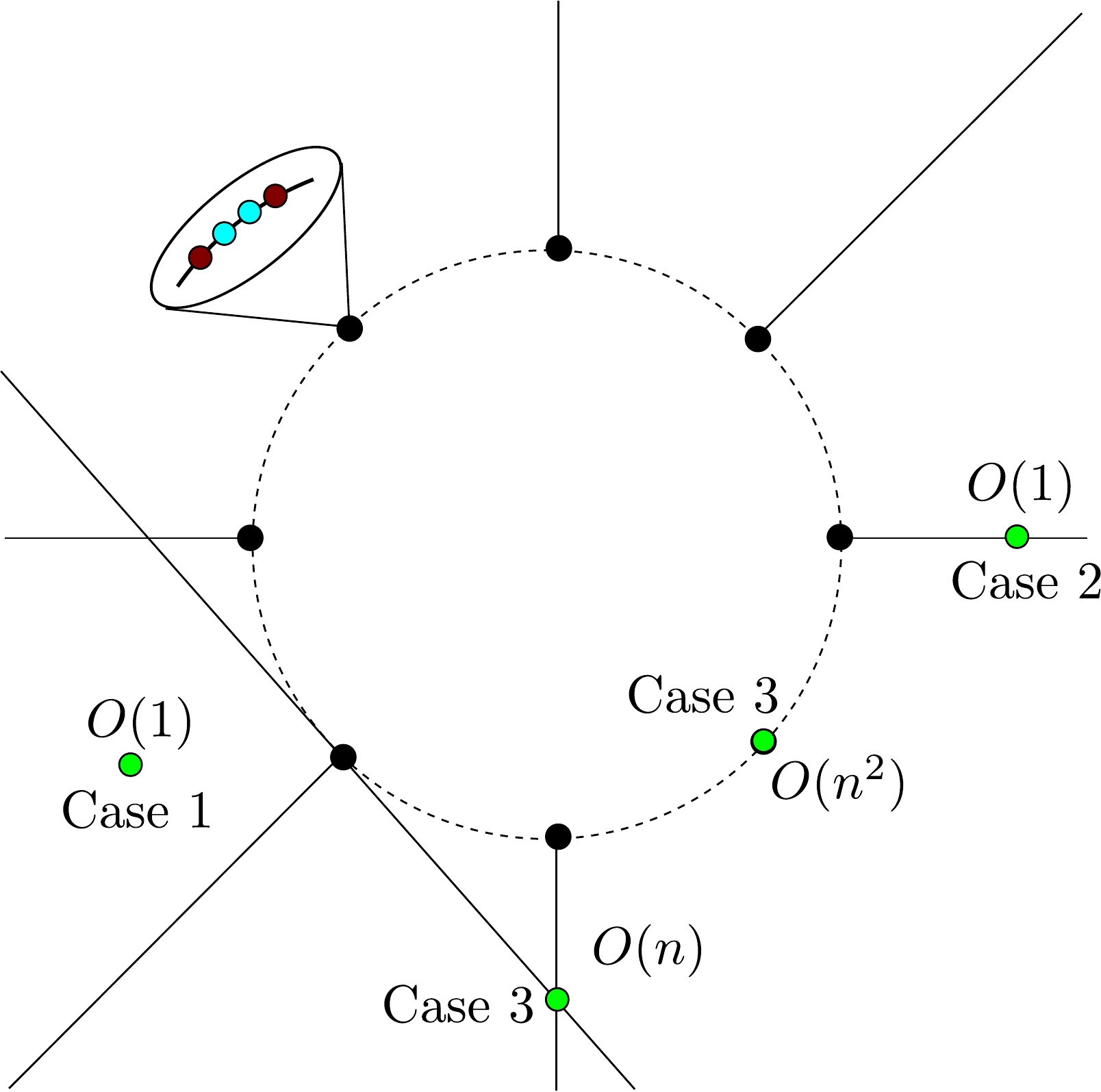}
  \end{center}
    \caption{A bi-colored set with $O(n^2)$ different
colored radial orderings.}\label{fig:colorfew}
\end{figure}

\begin{theorem} \label{thm:gcol_upper}
$\col{g}(n) \le O(n^2)$.
\end{theorem}
\begin{proof}
Recall that we are assuming that $n$ is even. 
We employ a similar technique as in the proof of Theorem \ref{thm:fcol_lower}.
We start with a set $S'$ of $n/2$ points,
placed almost evenly in the unit circle. All the points
of $S'$ have the same color and thus
the colored radial orderings of $S'$ are all
equivalent. Afterwards, we replace each
point of $S'$ with a symmetric pattern
of red and blue points. This is done in such a way that the new number of distinct colored radial
orderings increases at most to $O(n^2)$.

Let then $X$ be a set of $n/2$ points placed evenly on the unit
circle centered at the origin.
Explicitly, $X:=\{(\cos(\frac{4 \pi}{n} \cdot i),\sin(\frac{4 \pi}{n} \cdot i))\ |i=1,2,\dots,n/2\}$.
We choose $S'$ to be a set
of $n/2$ points (of the same color) each arbitrarily close to a distinct
point of  $X$ so that $S'$ is in strong general position. 
Let $L'$ be the set of lines passing through every pair of points
in $S'$ together with the lines passing through every point $p$ of $S'$ and
tangent to the circle centered at the origin and passing through $p$.
Let $A(S')$ be the line arrangement generated by $L'$.

Let $\delta>0$. We form a new set $S_\delta$, by replacing
 each point $p$ of $S'$ with a pattern of four
``red,blue,blue,red'' points, placed clockwise consecutively at the 
same distance as $p$ from the origin. The first (red) point is placed
at $p$, the next three points are placed at distance $\delta$
from the previous point.  Note that for
some small enough value $\delta'$, all line arrangements
$A(S_\delta)$ with $\delta < \delta'$ are combinatorially equivalent.
From now on assume that $\delta < \delta'$ and let $\delta$ tend to $0$.
The set of lines passing through each pair of points in $S_\delta$
tends to $L'$. Likewise, $A(S_\delta)$ tends to $A(S')$.
By this we mean that for each element $C$ (vertex, edge or cell)
of $A(S_\delta)$ there is an element $D$ of $A(S')$, so that
for each open set $\mathcal{O}$ containing $D$ there is a small enough value
for  $\delta$ 
so that $C$ is contained in $\mathcal{O}$.

We calculate an upper bound on the number of different
colored radial orderings of $S_{\delta}$ by considering the
colored radial orderings with respect to points
in the interior of each cell of $A(S_{\delta})$.
Let $C$ be a cell of $A(S_{\delta})$ and $q$ be 
a point in the interior of $C$.

There are three different cases according
to the limit of $C$ in $A(S')$ (refer to Figure~\ref{fig:colorfew}):

\begin{itemize}
\item {\bf Case 1.} $C$ tends to a cell of $A(S')$.\\
In this case, every pattern
replacing a point of $S'$ will
appear consecutively in the colored radial
ordering around $q$. Moreover, by
the symmetry of the patterns they are 
all ``red,blue,blue,red". Thus in this
case there is only one possible colored
radial ordering.

\item{\bf Case 2.} $C$ tends to an edge of $A(S')$.\\
We distinguish two sub-cases, whether
the edge is contained in one of the lines
passing through two points $p_i$ and $p_j$ of $S'$ or whether it
is contained in one of the tangent lines passing through a point $p_k$ of $S'$. 

In the first sub-case, the patterns at points different
from $p_i$ and $p_j$ will appear consecutively
in the colored radial ordering. 
However, at $p_i$ and $p_j$, the points in the patterns will appear together but 
intermixed. Since there are only
$8$ points involved, there is only a constant number
of  ways in which this can happen.
The second sub-case is similar.

\item{\bf Case 3.} $C$ tends to a vertex of $A(S')$.\\
In this case, $C$ may tend to a vertex
that is a point of $S'$ or the intersection
of two lines $\ell_1$ and $\ell_2$ of $L'$.

Suppose that $C$
tends to a point $p_i$ of $S'$.
 Then $C$ must be bounded by lines passing through
points of the pattern at $p_i$ and points of $S_{\delta}$. Note
that in total there are at most $8n$ such lines and thus
there are at most $O(n^2)$ such cells.
Moreover, since we may assume that $S'$
is arbitrarily close to $X$, our
analysis does not depend on the choice of $p_i$. Hence,  the cells tending 
to any other point of $S'$ will induce
the same set of colored radial orderings.

Suppose that $C$ tends to the intersection
point of two lines $\ell_1$ and $\ell_2$ of $L'$. These lines
may be defined by two points or one point of $S'$.
In both situations, we have that the patterns
at the points of $S'$ defining each line  appear
together but intermixed. The patterns at 
any other points will appear consecutively.
Since there is at most four points
defining $\ell_1$ and $\ell_2$, the number of 
ways in which their respective pattern points
can appear is at most a constant.
The only thing left to consider is the $O(n)$ number
of ways in which the 
patterns at the points defining $\ell_1$ can 
appear with respect to those of $\ell_2$. 

There are 
at most $O(n^2)$ distinct colored radial orderings in both subcases.
\end{itemize}

Note that the total number of distinct colored radial orderings
of $S_{\delta}$ is at most $O(n^2)$. Therefore by setting $S:=S_{\delta}$,
the result follows. 
\end{proof}

We now give  a linear lower bound for the number of colored radial orderings of $S$.
Some notation is required.
For a given radial ordering $\sigma$ of $S$, let
$\sigma(i)$ be its $(i+1)$-element.
Thus, two radial orderings $\sigma$ and $\rho$ are
equivalent whenever there exists a natural number $j$ such
that $\sigma(i)=\rho(i+j)$ for all $i$ 
(where addition is taken modulo $2n$); they
are equivalent as colored radial orderings
when the color of $\sigma(i)$ is equal
to the color of $\rho(i+j)$.

\begin{figure}[t]
  \begin{center}
    \includegraphics[width=0.9\textwidth]{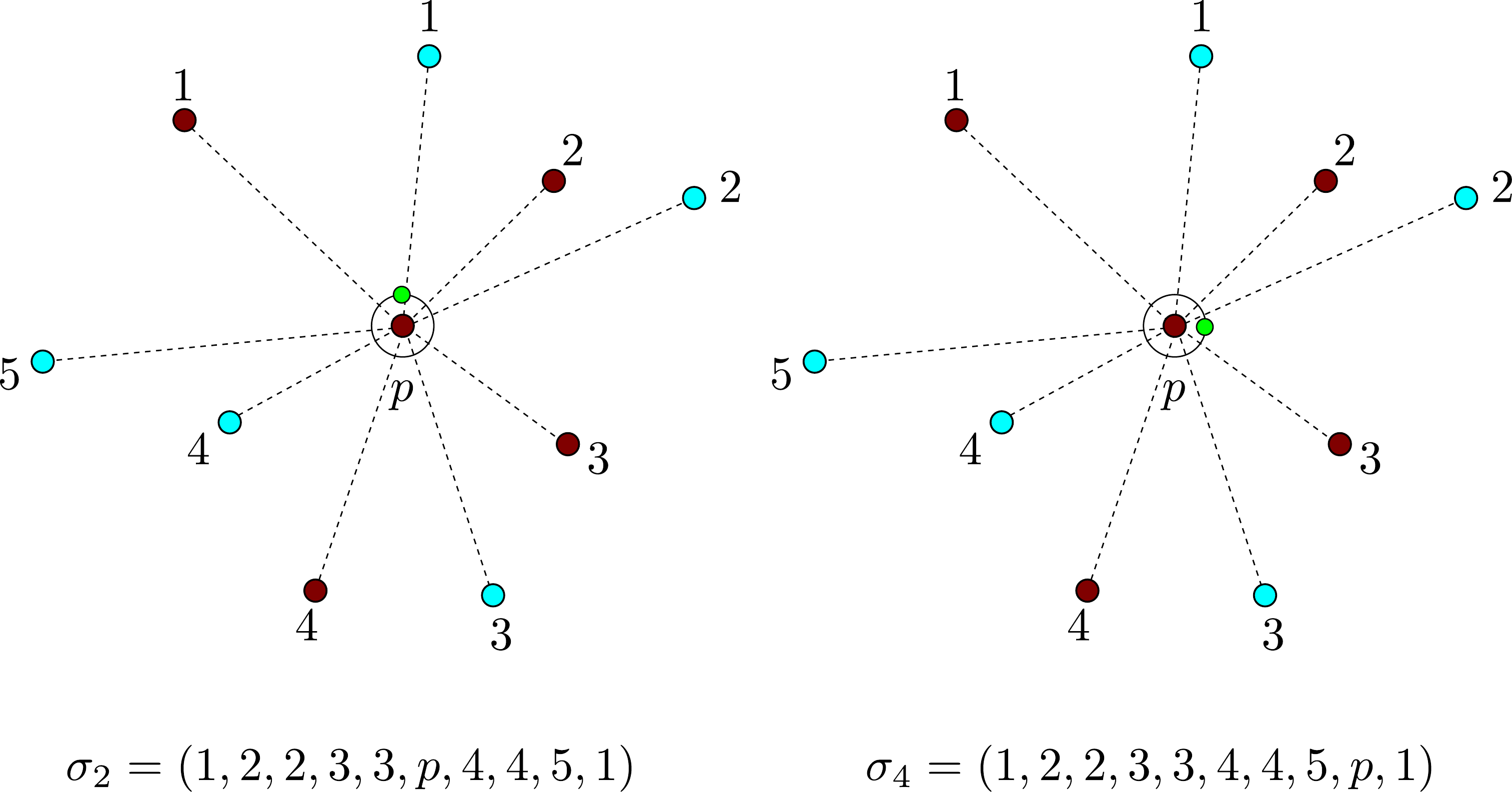}
  \end{center}
    \caption{A walk showing a linear lower bound on the number of colored
radial orderings.}\label{fig:lowercolor}
\end{figure}

\begin{theorem} \label{thm:gcol_lower}
$\col{g}(n) \ge n$.
\end{theorem}
\begin{proof}
Throughout the proof we use both \emph{colored} and \emph{non-colored}
radial orderings. In each instance we explicitly mention
to which of the two types of radial orderings we are
referring to. To obtain the claimed lower bound, 
we show a walk in which $n$
distinct \emph{colored} radial orderings are seen.
First we choose a red point $p$ of $S$ and let $\mathcal{C}$
be a circle centered at $p$. After that, we walk once
clockwise around $\mathcal{C}$. We choose $\mathcal{C}$ to
be small enough so that the only half-lines crossed
in the walk are those involving $p$ (see Figure~\ref{fig:lowercolor}). 

Consider the \emph{non-colored}  radial orderings
seen in this walk. Note that since
$\mathcal{C}$ does not cross any half-line defined
by points of $S \setminus \{p\}$, the points
of $S \setminus \{p\}$ remain fixed in these \emph{non-colored} radial
orderings. The only point that changes position
is $p$; it moves counter clockwise, transposing
an element of $S \setminus \{p\}$ every
time a half-line is crossed. We prove that every time that $p$ transposes
a blue point, a new different \emph{colored} radial ordering is seen.

Since in the walk around $\mathcal{C}$, the blue
points of $S$ always appear in the same
\emph{non-colored} radial ordering, we assume
that all the \emph{non-colored} radial orderings of $S$
as seen from points in $\mathcal{C}$ are
written starting at the same blue point.
Among these radial orderings, for each $k=0,\dots ,n-1$,
 let $\sigma_k$ be the radial ordering
in which $p$ is just after
the $(k+1)$-th blue point; refer to Figure~\ref{fig:lowercolor}. We show
that all the \emph{colored} radial orderings associated
to these $n$ \emph{non-colored} radial
orderings are distinct.

Let then $\sigma_k$ and $\sigma_l$ be two
such \emph{non-colored} radial orderings encountered at 
points $q_1$ and $q_2$ in the walk, respectively. 
Suppose that $\sigma_k$ and $\sigma_l$ are
equivalent as \emph{colored} radial orderings.
 Then there  exist a fixed natural number $j$ such that for all $i=0,\dots, 2n-1$, 
the color of $\sigma_k(i)$ is equal to 
the color of $\sigma_l(i+j)$. 

We now define a directed graph that captures
the relationship between $\sigma_k$ and 
$\sigma_l$; we employ the structure
of this graph to conclude that $k$ 
must equal $l$. 
Let $G$ be the directed graph whose vertex
set is $S$ and in which, for all $i=0,\dots, 2n-1$
there is an arc from $\sigma_k(i)$ to $\sigma_l(i+j)$; see
Figure~\ref{fig:cycles}. 
Note that every vertex
in $G$ has indegree and outdegree
equal to one. Therefore, $G$ is the union of pairwise disjoint 
directed cycles of points
of the same color. 

Let $\Gamma$ be the cycle containing $p$ and
let $S':=S\setminus V(\Gamma)$. 
Let $\rho_1$ and $\rho_2$ be the \emph{non-colored} radial
orderings of $S'$ as seen from $q_1$ and
$q_2$, respectively. We make
the additional assumption that
$\rho_1$ is written starting
at $\sigma_k(0)$ while
$\rho_2$ is written starting
at $\sigma_l(2n-j)$ (note that
being blue, these points are not in $\Gamma$).  Since in
particular $p $ is not in $S'$, these 
radial orderings are equivalent.
Therefore, there exists a fixed natural number $j'$ such that 
 $\rho_1(i)=\rho_2(i+j')$,
for all $i$. Since
$S'$ comes from removing the points
of $\Gamma$, $\rho_1$ can be formed
by removing the elements of $\Gamma$ from
$\sigma_k$ and $\rho_2$ can be formed
by writing $\sigma_l$ starting
at $\sigma_l(2n-j)$ and then removing
the elements of $\Gamma$ (See Figure~\ref{fig:cycles}). Thus the color of $\rho_1(i)$
is equal to the color of $\rho_2(i)$ for all $i$.

Let $G'$ be the directed graph whose vertex set
is $S'$ and in which there is an arc
from $\rho_1(i)$ to $\rho_2(i)$. As before,
every vertex in $G'$ has indegree and
outdegree equal to one. Therefore $G'$
 is the union of disjoint cycles of
vertices of the same color (in fact
$G'$ is the subgraph of $G$ induced
by $S'$). Since
$\rho_2$ is just a ``shift'' of $j'$ places to the right of $\rho_1$,
all of these cycles have the same
length $m$. Therefore, both
the number of red and blue
points in $S'$ are multiples
of $m$. This implies that
the number of vertices in
$\Gamma$ is also a multiple of $m$.

Let $r \cdot m$ be the length
of $\Gamma$, since $\Gamma$ is not empty, then $r \ge 1$.
Assume that $\Gamma$, starting
from $p$ is given by
$(p=v_1,v_2,\dots,v_m,\dots,v_{2m},\dots,v_{rm})$.
Let $b_k$ be the $(k+1)$-th blue point and
$\Gamma':=(b_k=u_1,\dots,u_m)$ be 
the cycle in $G$ containing $b_k$.
Consider the following sequence
of pairs of vertices $(u_1,v_1),(u_2,v_2),\dots,(u_m,v_m)$.
Note that in $\sigma_k$, the point  $v_1=p$ is just after  the point $u_1=b_k$;
afterwards, for $2 \le i \le m$, the point $v_i$ is just after
the point $u_i$ in both $\sigma_l$ and $\sigma_k$. 
(Recall that the order of $S\setminus \{p\}$ in $\sigma_k$ and $\sigma_l$
is the same.) Suppose that
$r>1$, then the point $v_{m+1}$ is just after $u_1$ in $\sigma_l$ while
in $\sigma_k$ it is just after the point $v_1$(which is equal to $p$).
Consider now the  following
sequence of vertices $(v_1,v_{m+1}),(v_2,v_{m+2}),\dots,(v_{(r-1)m+1},v_{rm+1}=p)$.
From the same arguments as before, for $m+2 \le i \le rm$, the point $v_i$ is just after
the point $v_{i-m}$ in both $\sigma_k$ and $\sigma_l$. 
For $i=rm+1$, the point $v_{rm+1}$ is just 
after the point $v_{(r-1)m+1}$ in $\sigma_l$, but $v_{rm+1}=p$
and $v_{(r-1)m+1}$ is red, a contradiction since
$p$ is just after a blue point in $\sigma_l$. 
Thus $r=1$. This implies that $\sigma_k=\sigma_l$, since $v_1$ 
(which is equal to $p$) is after
$u_1$ (which is equal to $b_k$) in both $\sigma_l$ and
$\sigma_k$. 
\end{proof} 

\begin{figure}[t]
  \begin{center}
    \includegraphics[width=0.9\textwidth]{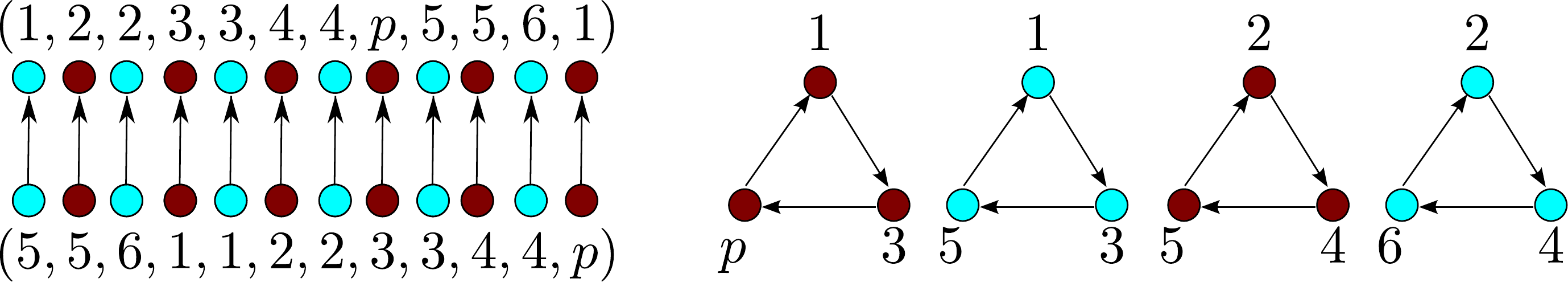}
  \end{center}
    \caption{Two equivalent colored radial orderings and their
corresponding graph.}\label{fig:cycles}
\end{figure}

\section{Conclusions}\label{sec:con}

We proved an upper bound of $O(n^4)$ and
a lower bound of $\Omega(n^{3})$ on
the number of radial orderings that
every set of $n$ points in strong general position in the plane must have.
The upper bound was first given in \cite{poly}. As a corollary
in the same paper it was noted that every set of $n$ points
in the plane contains $O(n^4)$ different
star shaped polygonizations.  Our result
implies that every set of $n$ points in strong general
position in the plane has $\Omega(n^{3})$ different
star shaped polygonizations. We leave the closing
of this gap as an open problem:

\begin{conjecture}
$g(n)=\Theta(n^4)$.
\end{conjecture}

We used the assumption that $S$ is in strong general
position heavily on the proof of the lower bound on $g(n)$.
However, we believe that it is not needed and that
only general position (no three collinear points)
is sufficient.

\begin{conjecture}
If $S$ is a set of $n$ points in general position in the plane, 
then it has at least $\Omega(n^4)$ distinct radial orderings.
\end{conjecture}

For colored point sets the situation is far more intriguing,
here we have been able to prove that such a gap exists.
Mainly that there are bi-colored sets of $2n$ points
with $\Theta(n^4)$ colored radial orderings and 
sets with  only $\Theta(n^2)$.
The best lower bound we have been able to provide is of $\Omega(n)$.
We make the following conjecture.

\begin{conjecture}
$\col{g}(n) = \Theta(n^2)$.
\end{conjecture}

Note that we used the assumption that  $n$ is even 
heavily in the proof of Theorem \ref{thm:gcol_upper}.
In fact, it can be shown that the number of colored radial
orderings may increase to $\Theta(n^3)$ if a red and a blue
point are added to the construction in the proof
of Theorem \ref{thm:gcol_upper}. Also in
the proof of Theorem \ref{thm:gcol_lower} we
relied on the fact that the number of red points
equals the number of blue points. It is possible
to construct a set of $n$ red and $n-1$ blue points
such that a walk like the one described in Theorem
\ref{thm:gcol_lower} yields only one colored
radial ordering. It may be the case that the bounds
given in Theorems \ref{thm:gcol_upper} and
\ref{thm:gcol_lower} no longer hold when either one of these
two hypothesis is dropped.
\\

{\bf Acknowledgments.}
The problems studied here, were introduced and partially solved during a stay at Universidad
de la Habana, Cuba, January 2010.
The authors would like to thank our host Carlos Ochoa during
this visit, and Clemens Huemer; Merce Claverol and David Wood  for helpful comments.

\bibliographystyle{plain}
\bibliography{orderingsbib}

\end{document}